\documentclass[12pt,reqno]{amsart}
\usepackage{graphicx,amscd, color,amsmath,amsfonts,amssymb,geometry,amssymb, xfrac,nicefrac}
\usepackage[initials]{amsrefs}

\newtheorem{theorem}{Theorem}[section]

\newtheorem{proposition}[theorem]{Proposition}
\newtheorem{corollary}[theorem]{Corollary}

\newtheorem{definition}[theorem]{Definition}
\numberwithin{equation}{section}

\begin{document}

\title{ Inequalities for the Schmidt Number of Bipartite States}

\author[Cariello ]{D. Cariello}

\address{Faculdade de Matem\'atica, \newline\indent Universidade Federal de Uberl\^{a}ndia, \newline\indent 38.400-902 Ð Uberl\^{a}ndia, Brazil.}\email{dcariello@ufu.br}


\subjclass[2010]{}

\begin{abstract}In this short note we show two completely opposite methods of constructing entangled states. Given a bipartite state $\gamma\in M_k\otimes M_k$, define $\gamma_S=(Id+F)\gamma (Id+F)$, $\gamma_A=(Id-F)\gamma(Id-F)$, where  $F\in M_k\otimes M_k$ is the flip operator.  In the first  method, entanglement is a consequence of the  inequality $\text{rank}(\gamma_S)<\sqrt{\text{rank}(\gamma_A)}$. 
In the second method, there is no correlation between $\gamma_S$ and $\gamma_A$. These two methods show how diverse is quantum entanglement.
 
We prove that any bipartite state $\gamma\in M_k\otimes M_k$ satisfies 
 $$\displaystyle SN(\gamma)\geq\max \left\{ \frac{\text{rank}(\gamma_L)}{\text{rank}(\gamma)}, \frac{\text{rank}(\gamma_R)}{\text{rank}(\gamma)}, \frac{SN(\gamma_S)}{2}, \frac{SN(\gamma_A)}{2} \right\},$$ 
where $SN(\gamma)$ stands for the Schmidt number of $\gamma$ and $\gamma_L,\gamma_R$ are the marginal states of $\gamma$.

We also present a family of PPT states in $M_k\otimes M_k$, whose members have Schmidt number equal to $n$, for any given  $1\leq n\leq \left\lceil\frac{k-1}{2}\right\rceil$. This is a new contribution to the open problem of finding the best possible Schmidt number for PPT states.  
\end{abstract}

\maketitle

\section{Introduction}

The separability problem in Quantum Information Theory asks for a deterministic criterion to distinguish the entangled states from the separable states \cite{Guhne}.  This problem is known to be a  hard problem even for bipartite mixed states \cites{gurvits2003, gurvits2004}. 

The Schmidt number of a state ($SN(\gamma)$ - Definition \ref{def1})  is a measure of how entangled a state is \cites{Terhal, Sperling2011}. If its Schmidt number is 1 then the state is separable. If its Schmidt number is greater than 1 then the state is entangled. A method to compute the Schmidt Number  is unknown.

Denote by  $M_k$ the set of complex matrices of order $k$.  The separability problem has been completely solved in $M_2\otimes M_2$. A state in $M_2\otimes M_2$  is separable if and only if it is positive under partial transposition or simply PPT (Definition \ref{def1}) \cites{peres,horodeckifamily}. Therefore, the Schmidt number of a PPT state in $M_2\otimes M_2$ is  equal to 1. Recently, the Schmidt number of every PPT state of $M_3\otimes M_3$ has been proved to be less or equal to 2 \cites{Yang, Chen}.

The authors of \cite{Marcus} left an open problem to determine the best possible Schmidt number for PPT states.   
They also presented  a construction of PPT states in $M_k\otimes M_k$ whose Schmidt numbers  are greater or equal to $\left\lceil\frac{k-1}{4}\right\rceil$.
This was the first explicit
example of a family of PPT states achieving a Schmidt number that scales linearly in the local dimension.  

We investigate this matter. We present an explicit construction of PPT states  in $M_k\otimes M_k$, whose Schmidt numbers are \textbf{equal} to $n$,  for any given  $1\leq n\leq \left\lceil\frac{k-1}{2}\right\rceil$. This is a new contribution to their open problem.

We manage to compute the Schmidt number of these PPT states using the following inequality  \begin{equation}\label{equation1}
SN(\gamma)\geq \max\left\{\frac{SN(\gamma_S)}{2}, \frac{SN(\gamma_A)}{2}\right\},
\end{equation} 
 where  $\gamma_S=(Id+F)\gamma (Id+F)$, $\gamma_A=(Id-F)\gamma(Id-F)$ and  $F\in M_k\otimes M_k$ is the flip operator $($i.e., $F(a\otimes b)=b\otimes a,$ for every $a,b\in\mathbb{C}^k)$. 
 
 We believe this is one of the simplest constructions of an entangled PPT state made so far.
 
 Another inequality that we present  here extends a result that  was previously known for separable states $($\cite[Theorem 1]{smolin}$)$ to every state of $ M_k\otimes M_m$.   Denote by $\gamma_L$ and $\gamma_R$ the marginal states of a state $\gamma\in M_k\otimes M_m$ (Definition \ref{def1}).
 
 We show that every state $\gamma$ of $M_k\otimes M_m$ satisfies\\
\begin{equation}\label{equation2}
\text{rank}(\gamma)SN(\gamma)\geq \max\{\text{rank}(\gamma_L), \text{rank}(\gamma_R)\}.
\end{equation} 

\vspace{0.3 cm}
We can use this inequality to obtain a lower bound for the Schmidt number of  low rank states.

 Next, through a series of very technical results, the author of \cite{Cariello_LAA} obtained the following lower bounds for the $\text{rank}(\gamma_S)$ of  any separable state $\gamma\in M_k\otimes M_k$   $$\text{rank}(\gamma_S)\geq \max\left\{\frac{r}{2}, \frac{2}{r} \text{rank}(\gamma_A)\right\},$$
where  $r$ is the marginal rank of  $\gamma+F\gamma F$.

These  inequalities can be combined into one inequality: \\
 \begin{center}
  $\displaystyle\text{rank}(\gamma_S)\geq  \frac{2}{r} \text{rank}(\gamma_A)\geq \frac{\text{rank}(\gamma_A)}{\text{rank}(\gamma_S)}.$
  \end{center} 
  \vspace{0.3 cm}
  
Hence,  $\text{rank}(\gamma_S)\geq \sqrt{\text{rank}(\gamma_A)}$ for every separable state $\gamma\in M_k\otimes M_k$.
Therefore,  if $\text{rank}(\gamma_S)<\sqrt{\text{rank}(\gamma_A)}$ then $\gamma$ is entangled. \\

Next, we can combine equations \ref{equation1} and \ref{equation2} in order to obtain
$$\displaystyle SN(\gamma)\geq \frac{\text{rank}((\gamma_S)_L)}{2\ \text{rank}(\gamma_S)}\text{ and  }\displaystyle SN(\gamma)\geq \frac{\text{rank}((\gamma_A)_L)}{2\ \text{rank}(\gamma_A)}.$$


We can easily create entangled states  by satisfying $\displaystyle\frac{\text{rank}((\gamma_S)_L)}{\text{rank}(\gamma_S)}>2$ or $ \displaystyle \frac{\text{rank}((\gamma_A)_L)}{\text{rank}(\gamma_A)}>2$ and no correlation between $\gamma_S$ and $\gamma_A$ is required. \\

These two methods of creating entangled states are completely opposite. One depends on a correlation between $\gamma_S, \gamma_A$ and the other does not.  They show how diverse is  quantum entanglement.\\

This paper is organized as follows. \begin{itemize}
\item In Section II, we prove that $SN(\gamma)\geq \max\left\{\frac{SN(\gamma_S)}{2}, \frac{SN(\gamma_A)}{2}\right\}$ (Proposition \ref{proposition1}) and we construct  a PPT  state whose Schmidt number is equal to $n$, for any given  $n\in\{1,\ldots,\left\lceil\frac{k-1}{2}\right\rceil\}$ (Proposition \ref{proposition2}). 
\item In Section III, we prove our main inequality $\text{rank}(\gamma)SN(\gamma)\geq \max\{\text{rank}(\gamma_L), \text{rank}(\gamma_R)\}$  (Theorem \ref{theoremnew0}) and two corollaries $\displaystyle SN(\gamma)\geq \frac{\text{rank}((\gamma_S)_L)}{2\ \text{rank}(\gamma_S)}$ and $ \displaystyle SN(\gamma)\geq \frac{\text{rank}((\gamma_A)_L)}{2\ \text{rank}(\gamma_A)}$ (Corollaries \ref{corollary1} and \ref{corollary2}). \\
\end{itemize}

Notation: Given $x\in\mathbb{R}$, define $\lceil x\rceil=\min\{n\in\mathbb{Z}, n \geq x \}$.   Identify  $M_k\otimes M_m\simeq M_{km}$ and  $\mathbb{C}^k\otimes \mathbb{C}^m\simeq \mathbb{C}^{km}$ via Kronecker product. Let us call a positive semidefinite Hermitian matrix of $M_{km}$ a (non-normalized bipartite) state of  $M_k\otimes M_m$. Let $\Im(\delta)$ denote the image of $\delta\in M_k\otimes M_m$ within  $\mathbb{C}^k\otimes \mathbb{C}^m$.  Given $w\in \mathbb{C}^k\otimes \mathbb{C}^m$ denote by $SR(w)$ its Schmidt rank $($or tensor rank$)$. Let  the trace of a matrix $A\in M_k$ be denoted by $tr(A)$. \\

\section{Preliminary Inequalities}

\vspace{0.3cm}

In this section we present two preliminary inequalities (Proposition \ref{proposition1}). They have independent interest as we can see in Proposition \ref{proposition2}. There we construct a family of PPT states in $M_k\otimes M_k$ whose members have Schmidt number equal to $n$, for any given  $1\leq n\leq \left\lceil\frac{k-1}{2}\right\rceil$. \\

 \begin{definition}\label{def1}
Given a state $\delta=\sum_{i=1}^nA_i\otimes B_i \in M_k\otimes M_m$, define  
 \begin{itemize}

 \item the Schmidt number of $\delta$  as \\$\displaystyle SN(\delta)=\min\left\{\max_{ j}\left\{SR(w_j) \right\},  \ \delta=\sum_{j=1}^mw_j\overline{w_j}^t \right\}$
 $($This minimum is taken over all decompositions of $\delta$ as $\sum_{j=1}^mw_j\overline{w_j}^t)$.
 \item   the partial transposition  of $\delta$ as $\delta^{\Gamma}=\sum_{i=1}^nA_i\otimes B_i^t$ .
Moreover, let us say that $\delta$ is positive under partial transposition or simply a PPT state  if and only if $\delta$ and $\delta^{\Gamma}$ are states.
\item the marginal states of $\delta$ as $\delta_L=\sum_{i=1}^nA_i tr(B_i)$ and $\delta_R=\sum_{i=1}^n B_i tr(A_i)$.
\\

\end{itemize}    
  \end{definition}

\begin{proposition}\label{proposition1} Every state  $\gamma\in M_k\otimes M_k$ satisfies $\displaystyle SN(\gamma)\geq \max\left\{\frac{SN(\gamma_S)}{2}, \frac{SN(\gamma_A)}{2}\right\}$.
\end{proposition} 
\begin{proof}
 By definition \ref{def1},  there is a subset $\{w_1,\ldots,w_n\}\subset\mathbb{C}^k\otimes \mathbb{C}^k$ such that
$\gamma=\sum_{i=1}^nw_i\overline{w_i}^t$  and  $SR(w_i)\leq SN(\gamma)$, for every $i$.\\

Therefore, $(Id\pm F)\gamma(I\pm F)=\sum_{i=1}^n v_i\overline{v_i}^t$, where $v_i=(Id\pm F)w_i$. Notice that, for every $i$, $$SR(v_i)=SR(w_i\pm Fw_i)\leq 2SR(w_i)\leq 2SN(\gamma).$$

Hence, $SN((Id\pm F)\gamma(I\pm F))\leq 2SN(\gamma)$.
\end{proof}

\vspace{0.3cm}

\begin{proposition}\label{proposition2}  Let    $v=\sum_{i=1}^na_i\otimes b_i$, where $\{a_1,\ldots,a_n,b_1,\ldots b_n\}$ is a linearly independent subset of $\mathbb{C}^k$.  Define $$\gamma=Id+F+\epsilon(v\overline{v}^t)\in M_k\otimes M_k.$$
\begin{enumerate}
\item For every $\epsilon>0$, $SN(\gamma)=n$. Notice that $1\leq n\leq \left\lceil\frac{k-1}{2}\right\rceil$. 
\item There is $\epsilon>0$ such that $\gamma$ is positive under partial transposition.
\end{enumerate}
\end{proposition}

\begin{proof}
(1) Notice that $ \gamma_A=(Id-F)\gamma(Id-F)=\epsilon(a\overline{a}^t)$, where $a=\sum_{i=1}^na_i\otimes b_i-b_i\otimes a_i$. 

Since 
$\{a_1,\ldots,a_n,b_1,\ldots b_n\}$ is  linearly independent then $SR(v)=n$ and $SR(a)=2n$. Hence, $$SN(\epsilon(v\overline{v}^t))=SR(v)=n\text{ and } SN(\gamma_A)=SR(a)=2n.$$

Thus, $\displaystyle SN(\gamma)\geq\frac{SN(\gamma_A)}{2}=n$, by  Proposition \ref{proposition1}.\\

Next, the separability of  $Id+F\in M_k\otimes M_k$ is a well known fact, therefore $SN(Id+F)=1$. \\

Finally, $SN(\gamma)\leq \max\{SN(Id+F),SN(\epsilon(v\overline{v}^t))\}= \max\{1,n\}=n$. 
Therefore, $SN(\gamma)=n$.\\

(2) Notice that $(Id+F)^{\Gamma}=Id+uu^t$, where $u=\sum_{i=1}^ k e_i\otimes e_i$ and $\{e_1,\ldots, e_k\}$ is the canonical basis of $\mathbb{C}^k$. So $(Id+F)^{\Gamma}$ is positive definite and, for a  small $\epsilon$,   $(Id+F)^{\Gamma}+\epsilon(v\overline{v}^t)^{\Gamma}$  is positive definite too.
\end{proof}

\section{Main Inequality}

In this section, we present our main result (Theorem \ref{theoremnew0}) and two corollaries (Corollaries \ref{corollary1} and \ref{corollary2}).\\

\begin{theorem}\label{theoremnew0}If $\gamma\in M_k\otimes M_m$ is a state then $\text{rank}(\gamma)SN(\gamma)\geq \max\{\text{rank}(\gamma_L),\text{rank}(\gamma_R)\}$.
\end{theorem}
\begin{proof}
The proof is an induction on $\text{rank}(\gamma)$. The case $\text{rank}(\gamma)=0$ is trivial. If  $\text{rank}(\gamma)=1$ then $SN(\gamma)=\text{rank}(\gamma_L)=\text{rank}(\gamma_R)$.\\

Let $\text{rank}(\gamma)>1$ and assume that this result is valid for states $\delta \in M_k\otimes M_m$ satisfying $\text{rank}(\delta)<\text{rank}(\gamma)$.\\

Since $\Im(\gamma)\subset\Im(\gamma_L\otimes\gamma_R)$ then $\gamma$ can be embedded in $M_{\text{rank}(\gamma_L)}\otimes M_{\text{rank}(\gamma_R)}$. The embedding does not change its rank or its Schmidt number. Thus, we can assume without loss of generality that $\text{rank}(\gamma_L)=k$ and $\text{rank}(\gamma_R)=m$.\\

Let $v\in\Im(\gamma)\setminus\{0\}$ be such that $SR(v)=SN(\gamma).$ \\

\begin{itemize}
\item If $k\geq m$ then choose $U\in M_k$ satisfying $\text{rank}(U)=k-SN(\gamma)$  and $(U\otimes Id)v=0$.\\ Define $\delta=(U\otimes Id)\gamma(U^*\otimes Id)$. Note that $\text{rank}(\delta)\leq \text{rank}(\gamma)-1$, since $\Im(\delta)\subset(U\otimes Id)(\Im(\gamma))$  and $(U\otimes Id)v=0$.\\

\item  If $k<m$ then choose $U\in M_m$ satisfying $\text{rank}(U)=m-SN(\gamma)$  and $(Id\otimes U)v=0$.\\ Define $\delta=(Id\otimes U)\gamma(Id\otimes U^*)$. Note that $\text{rank}(\delta)\leq \text{rank}(\gamma)-1$, since $\Im(\delta)\subset(Id\otimes U)(\Im(\gamma))$  and $(Id\otimes U)v=0$.\\
\end{itemize}

In any case, by induction hypothesis, $ \text{rank}(\delta)SN(\delta)\geq\max\{\text{rank}(\delta_L), \text{rank}(\delta_R)\}$.\\

\begin{itemize}
\item If $k\geq m$ then $\delta_L=U\gamma_LU^*$. Since $\gamma_L$ is positive definite then $\text{rank}(\delta_L)=\text{rank}(U)=k-SN(\gamma).$\\

\item  If $k<m$ then $\delta_R=U\gamma_RU^*$. Since $\gamma_R$ is positive definite then $\text{rank}(\delta_R)=\text{rank}(U)=m-SN(\gamma).$\\
\end{itemize}

Since $\text{rank}(\delta)\leq \text{rank}(\gamma)-1$ and $SN(\delta)\leq SN(\gamma)$ then\\

\begin{itemize}
\item$(\text{rank}(\gamma)-1)SN(\gamma)\geq k-SN(\gamma)$,  if $k\geq m$. Therefore, $\text{rank}(\gamma)SN(\gamma)\geq k$. \\

\item$(\text{rank}(\gamma)-1)SN(\gamma)\geq m-SN(\gamma)$,  if $k<m$. Therefore, $\text{rank}(\gamma)SN(\gamma)\geq m$. \\
\end{itemize}

The induction is complete.
\end{proof}
\vspace{0.3cm}

\begin{corollary}\label{corollary1} If $\gamma\in M_k\otimes M_k$ is a state then $\displaystyle SN(\gamma)\geq \frac{\text{rank}((\gamma_A)_L)}{2\ \text{rank}(\gamma_A)}$.
\end{corollary}

\begin{proof}

First, notice that $(\gamma_A)_L=(\gamma_A)_R$. Therefore, $\text{rank}((\gamma_A)_L)=\text{rank}((\gamma_A)_R)$.\\

Next, since $SN(\gamma_A)\leq 2SN(\gamma)$, by Proposition \ref{proposition1}, then  $\displaystyle\text{rank}(\gamma_A)SN(\gamma)\geq\frac{1}{2}\text{rank}(\gamma_A)SN(\gamma_A)\geq  \frac{\text{rank}((\gamma_A)_L)}{2}$, by Theorem \ref{theoremnew0}.
\end{proof}

\begin{corollary}\label{corollary2} If $\gamma\in M_k\otimes M_k$ is a state then $\displaystyle SN(\gamma)\geq \frac{\text{rank}((\gamma_S)_L)}{2\ \text{rank}(\gamma_S)}$.
\end{corollary}

\begin{proof}
First, notice that $(\gamma_S)_L=(\gamma_S)_R$. Therefore, $\text{rank}((\gamma_S)_L)=\text{rank}((\gamma_S)_R)$.\\

Since $SN(\gamma_S)\leq 2SN(\gamma)$, by Proposition \ref{proposition1}, then  $\displaystyle \text{rank}(\gamma_S)SN(\gamma)\geq\frac{1}{2}\text{rank}(\gamma_S)SN(\gamma_S)\geq  \frac{\text{rank}((\gamma_S)_L)}{2}$, by Theorem \ref{theoremnew0}.
\end{proof}

\vspace{0.3cm}

  \section*{Summary and Conclusion}

We presented an inequality that relates the marginal ranks of any bipartite state of $M_k\otimes M_m$ to its rank and its Schmidt number .  Using this inequality, we described a method of constructing entangled states which is not based on any correlation between $\text{rank}(\gamma_A)$ and  $\text{rank}(\gamma_S)$. 
This form of entanglement  differs completely from the entanglement derived from the inequality  $\text{rank}(\gamma_S)<\sqrt{\text{rank}(\gamma_A)}$.
We also constructed  a family of PPT states whose members have Schmidt number equal to $n$, for any given  $1\leq n\leq \left\lceil\frac{k-1}{2}\right\rceil$. This is a new contribution to the open problem of finding the best possible Schmidt number for PPT states.

\end{document}